\documentclass[12pt]{article}
\usepackage{amsmath,amssymb,amsthm,graphicx,mathrsfs,color,fancyhdr,fancybox,setspace}
\usepackage{placeins}
\usepackage{amsfonts}

\newtheorem{theorem}{Theorem}[section]
\newtheorem{example}{Example}
\newtheorem{problem}[theorem]{Problem}
\newtheorem{conjecture}[theorem]{Conjecture}

  \newtheorem{observation}[theorem]{Observation}
  
  \newtheorem{proposition}[theorem]{Proposition}
\newtheorem{defin}[theorem]{Definition}
\newtheorem{lemma}[theorem]{Lemma}
\newtheorem{corollary}[theorem]{Corollary}
\newtheorem{nt}{Note}

\begin{document}
\parindent 0cm
\parskip .3cm
\thispagestyle{empty}
\begin{center}
\Large
 {\bf Almost every graph is divergent under the biclique operator}
\vfill

\large
Marina Groshaus \footnote{Partially
  supported by UBACyT grant  20020100100754, PICT ANPCyT grant 2010-1970,
  CONICET PIP grant 11220100100310}\footnote{Partially
  supported by Math-Amsud project 14 Math 06}\\
 \normalsize
Universidad de Buenos Aires \\
Departamento de Computaci\'on \\
groshaus@dc.uba.ar \\
\vfill

\large
Andr\'{e} L.P. Guedes\footnotemark[2]\\
 \normalsize
Universidad de Buenos Aires / Universidade Federal do Paran\'{a} \\
Departamento de Computaci\'on / Departamento de Inform\'{a}tica \\
andre@inf.ufpr.br \\
\vfill

\large
Leandro Montero \\
 \normalsize
Universidad de Buenos Aires / Universit\'e Paris-Sud \\
Departamento de Computaci\'on / Laboratoire de Recherche en Informatique \\
lmontero@\{dc.uba.ar/lri.fr\} \\

\normalsize 
\vspace{.5cm}

ABSTRACT
\end{center}

\small A biclique of a graph $G$ is a maximal induced complete bipartite subgraph of $G$. The biclique graph of $G$ denoted by $KB(G)$, is the intersection 
graph of all the bicliques of $G$. The biclique graph can be thought as an operator between the class of all graphs. The iterated biclique graph of $G$ 
denoted by $KB^{k}(G)$, is the graph obtained by applying the biclique operator $k$ successive times to $G$. 
The associated problem is deciding whether an input graph converges, diverges or is periodic under the biclique operator when $k$ grows to infinity. 
All possible behaviors were characterized recently and an $O(n^4)$ algorithm for deciding the behavior of any graph under the biclique operator was also 
given. In this work we prove new structural results of biclique graphs. In particular, we prove that every false-twin-free graph with at least 
$13$ vertices is divergent. These results lead to a linear time algorithm to solve the same problem.

\normalsize

\newpage

\section{Introduction}
Intersection graphs of certain special subgraphs of a general graph have been studied 
extensively. For example, line graphs (intersection graphs of the edges of a graph), 
interval graphs (intersection of intervals of the real line), clique graphs (intersection of cliques of a graph), etc 
\cite{BoothLuekerJCSS1976,BrandstadtLeSpinrad1999,EscalanteAMSUH1973,FulkersonGrossPJM1965, GavrilJCTSB1974, LehotJA1974,McKeeMcMorris1999}.

The \emph{clique graph} of $G$ denoted by $K(G)$, is the intersection graph of the family of all maximal cliques of $G$. Clique graphs were introduced by Hamelink in \cite{HamelinkJCT1968} and characterized by Roberts and Spencer in  \cite{RobertsSpencerJCTSB1971}. 
The computational complexity of the recognition problem of clique graphs had been open for more that 40 years. In  \cite{Alc'onFariaFigueiredoGutierrez2006} they proved that clique graph recognition problem is NP-complete.

The clique graph can be thought as an operator between the class of all graphs. The \emph{iterated clique graph} $K^{k}(G)$ is the graph obtained by applying the clique operator $k$ successive times ($K^0(G)=G$). Then $K$ is called \emph{clique operator} and it was introduced by Hedetniemi and Slater in \cite{HedetniemiSlater1972}. 
Much work has been done on the scope of the clique operator looking at the different possible behaviors. 
The associated problem is deciding whether an input graph converges, diverges or is periodic under the clique operator when $k$ grows to infinity. 
In general it is not clear that the problem is decidable.  
However, partial characterizations have been given for convergent, divergent and periodic graphs restricted to some classes of graphs. Some of these lead to polynomial time recognition algorithms.
For the clique-Helly graph class, graphs which converge to the trivial graph have been characterized in~\cite{BandeltPrisnerJCTSB1991}.
Cographs, $P_4$-tidy graphs, and circular-arc graphs  are examples of classes where the different behaviors are 
characterized~\cite{MelloMorganaLiveraniDAM2006,LarrionMelloMorganaNeumann-LaraPizanaDM2004}. Divergent graphs were also considered. For example in 
\cite{Neumann1981}, families of divergent graphs are shown. Periodic graphs were studied in
\cite{EscalanteAMSUH1973,LarrionNeumann-LaraPizanaDM2002}. In particular it is proved that for every integer $i$, there exist periodic graphs with period $i$ and also convergent graphs which converge in $i$ steps. More results about iterated clique graph can be found in 
\cite{LarrionPizanaVillarroel-FloresDM2008,Frias-ArmentaNeumann-LaraPizanaDM2004,LarrionNeumann-LaraGC1997,LarrionNeumann-LaraDM1999,
LarrionNeumann-LaraDM2000,PizanaDM2003}.

A biclique is a maximal bipartite complete induced subgraph. Bicliques have applications in various fields, for example biology: protein-protein interaction networks~\cite{Bu01052003}, 
social networks: web community discovery~\cite{Kumar}, genetics~\cite{Atluri}, medicine~\cite{Niranjan}, information theory~\cite{Haemers200156}, etc. 
More applications (including some of these) can be found in~\cite{blablamec}.

The \emph{biclique graph} of a graph $G$ denoted by $KB(G)$, is the intersection graph of the family of all maximal bicliques of $G$. 
It was defined and characterized in~\cite{GroshausSzwarcfiterJGT2010}. However no polynomial time algorithm is known for recognizing biclique graphs. 
As for clique graphs, the biclique graph construction can be viewed as an operator between the class of graphs. 

The \emph{iterated biclique graph} $KB^{k}(G)$ is the graph obtained by applying to $G$ the biclique operator $KB$ $k$ times iteratively. It was introduced in~\cite{marinayo} and all possible behaviors were characterized.
It was proven that a graph $G$ is either divergent or convergent
but it is never periodic (with period bigger than $1$). In addition, general characterizations for 
convergent and divergent graphs were given. These results were based on the fact that if a graph $G$ contains a clique of size at least $5$, 
then $KB(G)$ or $KB^2(G)$ contains a clique of larger size. Therefore, in that case $G$ diverges. Similarly if $G$ contains the $gem$ or the $rocket$ 
graphs as an induced subgraph, then $KB(G)$ 
contains a clique of size $5$, and again $G$ diverges. Otherwise it was shown that after removing false-twin vertices of $KB(G)$, the resulting graph is a clique on at most $4$ vertices, in which case $G$ converges. Moreover, it was proved that if a graph $G$ converges, it converges to the graphs 
$K_1$ or $K_3$, and it does so in at most $3$ steps. 
These characterizations leaded to an $O(n^4)$ time algorithm for recognizing convergent or divergent graphs under the biclique operator. 

In this work we show new results that lead to a linear time algorithm to
solve the same problem. We study conditions for a graph to contain a $K_5$, a $C_5$, 
a $butterfly$, a $gem$ or a $rocket$ (see Figure~\ref{grafitos}) as induced subgraphs so we can decide
divergence (since $K_5 \subseteq KB(C_5),KB(butterfly),KB(gem),$ $KB(rocket)$). 
First we prove that if $G$ has at least $7$ bicliques then it diverges. 
Then, we show that every false-twin-free graph with at least $13$ vertices has at least $7$ bicliques, and therefore diverges.
Since adding false-twins to a graph does not change its $KB$ behavior, then the linear algorithm is based on the deletion of false-twin vertices of 
the graph and looking at the size of the remaining graph.  

\begin{figure}[ht!]
	\centering
	\includegraphics[scale=.2]{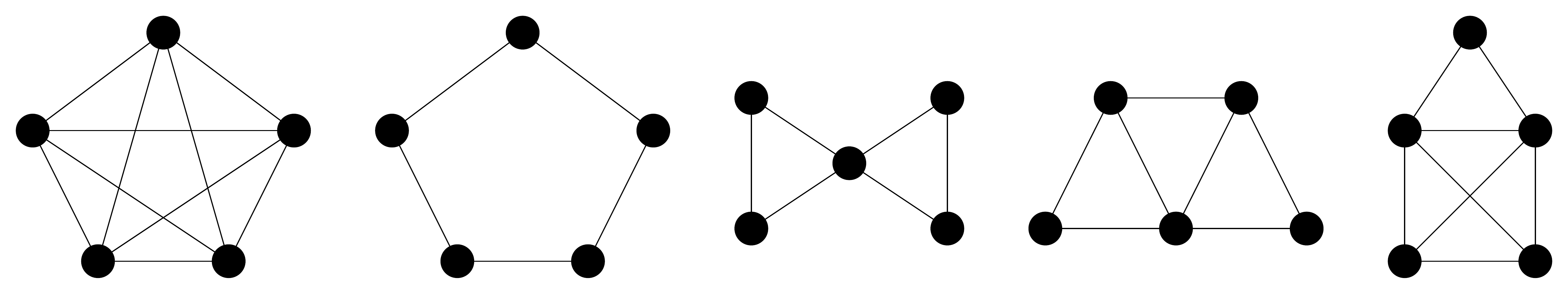}
	\caption{Graphs $K_5$, $C_5$, $butterfly$, $gem$ and $rocket$, respectively.}
	\label{grafitos}
\end{figure}

It is worth to mention that these results are indeed very different from the ones known for the clique operator, for which it is still an open problem to know the computational complexity of deciding the behavior of a graph under the clique operator. 

This work is the full version of a previous extended abstract  published in~\cite{marinayoENDM}. It is organized as follows. In Section $2$ the notation is given. Section $3$ contains some preliminary results that we will use later. In Section $4$ we prove that any graph with at least $7$ bicliques diverges, and that every graph with at least $13$ vertices with no false-twins vertices  contains at least $7$ bicliques. This leads to a linear time algorithm to decide convergence or divergence under the biclique operator.

\section{Notation and terminology}

Along the paper we restrict to undirected simple graphs. Let $G=(V,E)$ be a graph with vertex set $V(G)$ and edge set $E(G)$, and 
let $n=|V(G)|$ and $m=|E(G)|$. A \emph{subgraph} $G'$ of $G$ is a graph $G'=(V',E')$ where $V'\subseteq V$ and 
$E'\subseteq E$. A subgraph $G'=(V',E')$ of $G$ is \emph{induced} when for every pair of vertices $v,w \in G'$, $vw \in E'$ 
if and only if $vw \in E$. A graph $G$ is \emph{$H$-free} if it does not contain $H$ as an induced subgraph.
A graph $G=(V,E)$ is \emph{bipartite} when $V= U \cup W$, $U \cap W = \emptyset$ and $E \subseteq U \times W$. Say that $G$ is a 
\emph{complete graph} when every possible edge belongs to $E$. A complete graph of $n$ vertices is denoted $K_{n}$. A \emph{clique} of $G$ is a maximal complete induced subgraph while a \emph{biclique} is a maximal bipartite complete induced subgraph of $G$. The \emph{open neighborhood} of a 
vertex $v \in V(G)$ denoted $N(v)$, is the set of vertices adjacent to $v$ while the \emph{closed neighborhood} of $v$ denoted by 
$N[v]$, is $N(v) \cup \{v\}$. Two vertices $u$, $v$ are \emph{false-twins} if $N(u)=N(v)$. A vertex $v\in V(G)$ is \emph{universal} if it is 
adjacent to all of the other vertices in $V(G)$. A \emph{path} (\emph{cycle}) of $k$ vertices, denoted by $P_{k}$ ($C_k$), is a set of vertices 
$v_{1}v_{2}...v_{k} \in V(G)$ such that $v_{i} \neq v_{j}$ for all $1 \leq i \neq j \leq k$ and $v_{i}$ is adjacent to $v_{i+1}$ 
for all $1 \leq i \leq k-1$ (and $v_1$ is adjacent to $v_k$). A graph is \emph{connected} if there exists a path between each pair of vertices. 
We assume that all the graphs of this paper are connected.

A $rocket$ is a complete graph with $4$ vertices and a vertex adjacent to two of them. 
A $butterfly$ is the graph obtained by joining two copies of the $K_3$ with a common vertex.

Given a family of sets $\mathcal{H}$, the \emph{intersection graph} of $\mathcal{H}$ is a graph that has the members of 
$\mathcal{H}$ as vertices and there is an edge between two sets $E,F\in\mathcal{H}$ when $E$ and $F$ have non-empty intersection.

A graph $G$ is an \emph{intersection graph} if there exists a family of sets $\mathcal{H}$ such that $G$ is the intersection 
graph of $\mathcal{H}$. We remark that any graph is an intersection graph \cite{Szpilrajn-MarczewskiFM1945}. 

A family of sets $\mathcal{H}$ is \emph{mutually intersecting} if every pair of sets $E,F\in\mathcal{H}$ have non-empty intersection.

Let $F$ be any graph operator. Given a graph $G$, the \emph{iterated graph} under the operator $F$ is defined iteratively as 
follows: $F^{0}(G)=G$ and for $k\geq 1$, $F^{k}(G)=F^{k-1}(F(G))$. We say that a graph $G$ \emph{diverges} under the operator $F$ whenever 
$\lim_{k \rightarrow \infty}|V(F^{k}(G))|=\infty$. We say that a graph $G$ \emph{converges} under the operator $F$ whenever 
$F^{m+1}(G)=F^{m}(G)$ for some $m$, that is, $F^{k}(G)=F^{m}(G)$ for every $k \geq m$ and some $m$. We say that a graph $G$ is \emph{periodic} under 
the operator $F$ whenever $F^{k}(G)=F^{k+s}(G)$ for some $k,s$, $s \geq 2$. 

The \emph{iterated biclique graph} $KB^{k}(G)$ is the graph obtained by applying iteratively the biclique operator $k$ times to $G$. 

In the paper we will use the terms convergent or divergent meaning convergent or divergent under the biclique operator $KB$.

By convention we arbitrarily say that the trivial graph $K_1$ is convergent under the biclique operator (observe that this remark is needed 
since the graph $K_1$ does not contain bicliques).

\section{Preliminary results}\label{resultsantes}

We start with this easy observation.

\begin{observation}[\cite{marinayo}]\label{l3}
If $G$ is an induced subgraph of $H$, then $KB(G)$ is a subgraph (not necessarily induced) of $KB(H)$.
\end{observation}

The following proposition is central in the characterization of convergent and divergent graphs under the biclique operator. Basically, it shows that if a graph contains a big complete subgraph, it is going to grow in one or two steps of $KB$.

\begin{proposition}[\cite{marinayo}]
Let $G$ be a graph that contains $K_{n}$ as a subgraph, for some $n \geq 4$. Then, $K_{2n-4} \subseteq KB(G)$ or 
$K_{(n-2)(n-3)} \subseteq KB^{2}(G)$.
\end{proposition}

Next theorem characterizes the behavior of a graph under the biclique operator.

\begin{theorem}[\cite{marinayo}] \label{divergencia}
If $KB(G)$ contains either $K_5$ or the $gem$ or the $rocket$ as an induced subgraph, then $G$ is divergent. 
Otherwise, $G$  converges to $K_1$ or $K_3$ in at most 3 steps.
\end{theorem}

Notice that differently than the clique operator, a graph is never periodic under the biclique operator (with period bigger than 1). 
We remark the importance of the graph $K_5$ to decide the behavior of a graph under the biclique operator since we have that 
$KB(gem) = K_5$ and $K_5 \subseteq KB(rocket)$.

Observe that as proved in~\cite{marinayo}, the biclique graph does not change by the deletion or addition of false-twin vertices since each pair of 
false-twins belongs to exactly the same set of bicliques. That is, for any graph $G$, $KB(G)=KB(G-\{v\})$ for any false-twin vertex $v$. It follows that the behavior of a graph under $KB$ does not change either. Therefore we focus our study on false-twins-free graphs. For that we need the following definition used in~\cite{marinayo}.

Consider all maximal sets of false-twin vertices $Z_1,...Z_k$ and let $\{z_{1},z_{2},...,z_{k}\}$ be the set of \emph{representative vertices} such that $z_i\in Z_i$. The graph obtained by the deletion of all vertices of $Z_i- \{z_i\}$ for $i=1,...,k$, is denoted $Tw(G)$. Observe that $Tw(G)$ has no false-twin vertices.

Using $Tw(G)$, as a corollary of Theorem~\ref{divergencia}, the next useful result was obtained.

\begin{corollary}[\cite{marinayo}]\label{contraccion}
A graph $G$ is convergent if and only if $Tw(KB(G))$ has at most four vertices. Moreover, $Tw(KB(G))=K_{n}$ for $n=1,...,4$.
\end{corollary}

We recall that the number of bicliques of a graph may be exponential in the number of its vertices~\cite{PrisnerC2000}. However, if some vertex of a 
graph $G$ lies in five bicliques, then $KB(G)$ contains a $K_5$ thus $G$ diverges. If every vertex of $G$ belong to at most four bicliques, then $G$ has at 
most $2n$ bicliques. Therefore, since each biclique can be generated in $O(n^3 )$~\cite{DiasFigueiredoSzwarcfiterTCS2005,DiasFigueiredoSzwarcfiterDAM2007}, 
constructing $KB(G)$ takes $O(n^4)$. Building $Tw(KB(G))$ can be done in $O(n + m)$ time using the modular decomposition~\cite{modulardecomp}. From 
Corollary~\ref{contraccion}, if $Tw(KB(G))$ has at most four vertices, then $G$ is convergent, otherwise $G$ is divergent. Hence the overall algorithm 
runs in $O(n^4)$ time.

\section{Linear time algorithm}\label{algolinear}

In this section we give a linear time algorithm for deciding whether a given
graph is divergent or convergent under the biclique operator.

Motivated by Theorem~\ref{divergencia} and Corollary~\ref{contraccion}, we study structural properties of a graph $H$ to guarantee that 
its biclique graph $G=KB(H)$ contains $K_5$ and therefore $H$ diverges.

The following two lemmas answer that question.

\begin{lemma}\label{2gemelos}
Let $G=KB(H)$ for some graph $H$. Let $b_1$, $b_2$ be false-twin vertices of $G$ and
$B_1$, $B_2$ their associated bicliques in $H$. Suppose that there are no edges
between vertices of $B_1$ and vertices of $B_2$. Then there exists a vertex $v
\in H$ such that $v$ is adjacent to every vertex of $B_1$ and $B_2$.
Furthermore, $G$ contains a $K_5$ as induced subgraph.
\end{lemma}

\begin{proof} Let $b_1$, $b_2$ be false-twin vertices of $G$ and $B_1$, $B_2$ their associated bicliques in $H$, such that there are no edges between vertices of $B_1$ and vertices of $B_2$. Since $G$ is connected, take the shortest path from some vertex of $B_1$ to $B_2$. Let $w$ be the first vertex in the path such that $w \notin B_1$. Clearly, $w \notin B_2$. Let $v \in B_1$ be a vertex adjacent to $w$. 

First, suppose that there exists a vertex $x \in B_1$ such that $x$ is not adjacent to $w$. Consider the following alternatives:

\textbf{Case 1}: $xv \in E(H)$. Then $\{x,v,w\}$ is contained in some biclique $B$, $B \neq B_1$ and $B \neq B_2$, such that it does not intersect $B_2$ since there is no edge between $B_1$ and $B_2$. This is a contradiction since $b_1$ and $b_2$ are false-twin vertices. It follows that every vertex in $B_1$ not adjacent to $w$ is not adjacent to $v$.

\textbf{Case 2}: $xv \notin E(H)$. Then there exists a vertex $y \in B_1$ adjacent to $v$ and $x$. By case 1, $y$ must be adjacent to $w$. This is the same situation as previous case but considering $y$ instead of $v$ and the biclique containing $\{x,y,w\}$ instead of $\{x,v,w\}$. A contradiction.

We conclude that for all $x \in B_1$, $x$ is adjacent to $w$.

Now, the edge $vw$ is contained in a biclique $B$ that must intersect $B_2$ as $b_1,b_2$ are false-twin vertices of $G$. Since there are no edges between $B_1$ and $B_2$ there exists a vertex 
$z \in B_2 $ such that $z$ is adjacent to $w$. The same argument used for $v \in B_1$ and $w$ also holds for $z \in B_2$ and $w$. That is, for all 
$z\in B_2$, $z$ is adjacent to $w$.

\begin{figure}[ht!]
	\centering
	\includegraphics[trim=0 100 100 140, clip, scale=.4]{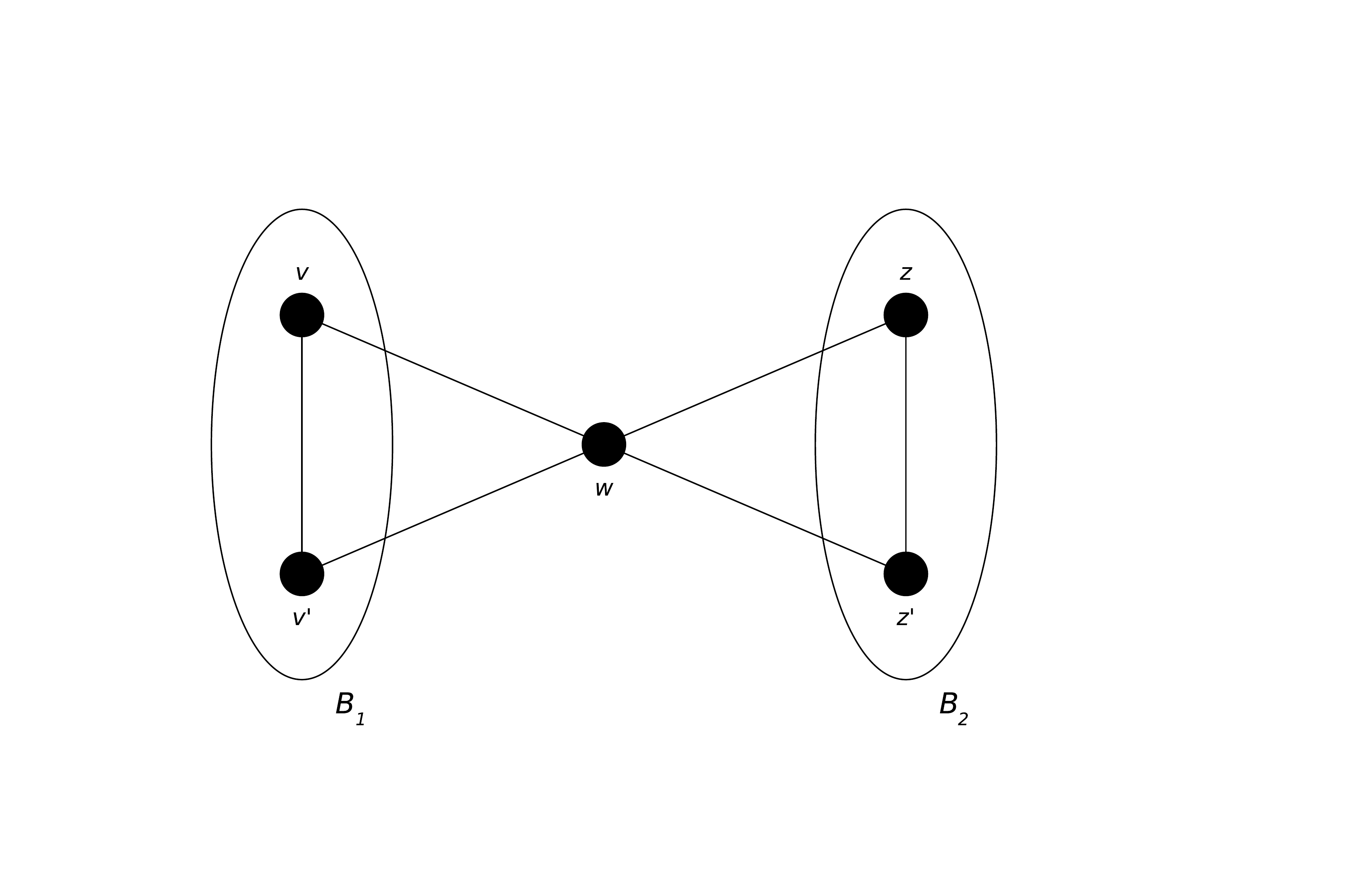}
	\caption{Bicliques $B_1$ and $B_2$ with $4$ new bicliques containing $w$.}
	\label{Figlemma1}
\end{figure}

Finally, let $v,v'$ be adjacent vertices in $B_1$ and let $z,z'$ be adjacent vertices in $B_2$. Since $v,v',z,z'$ are adjacent to $w$, then $\{v,w,z\}$, $\{v',w,z\}$, $\{v,w,z'\}$ and $\{v',w,z'\}$ are contained in four different bicliques $B_3$, $B_4$, $B_5$ and $B_6$ such that $B_i \neq B_j$, for $1 \leq i \neq j \leq 6$. As $B_i \cap B_j \neq \emptyset$, for $2 \leq i \neq j \leq 6$ (Fig.~\ref{Figlemma1}), $K_5$ is an induced subgraph of $G$.
\end{proof}

\begin{lemma}\label{3gemelos}
Let $G=KB(H)$ for some graph $H$. Let $b_1, b_2, b_3$ be false-twin vertices of $G$
and let $B_1, B_2, B_3$ be their associated bicliques in $H$. Suppose that for
any pair of bicliques $B_i, B_j$, $1 \leq i \neq j \leq 3$, there is an edge
between some vertex of $B_i$ and some vertex of $B_j$. Then, $K_5$ is an induced
subgraph of $G$.
\end{lemma}

\begin{proof}
Let $b_1, b_2, b_3$ be the false-twin vertices of $G$ and $B_1, B_2, B_3$ their associated
bicliques in $H$ such that for any pair of bicliques $B_i, B_j$, $1 \leq i \neq
j \leq 3$, there is an edge between some vertex of $B_i$ and some vertex of
$B_j$. We will show that $H$ contains either a $butterfly$, a $gem$, a $rocket$ or a $C_5$, or 
four mutually intersecting bicliques also intersecting with $B_1$, $B_2$ and $B_3$.
In any case we obtain a $K_5$ in $G$.
We have the following cases:

\textbf{Case 1}: There is a $K_3$ with one vertex in each biclique. Let $a\in B_1$, $b\in
B_2$, $c\in B_3$ be the $K_3$. Now $ab$, $ac$ and $bc$ are contained
in $3$ different bicliques of $H$. It is easy to see that none of $B_1$, $B_2$ or $B_3$ are bicliques isomorphic to $K_2$, otherwise 
they would not intersect the biclique containing the opposite edge of the $K_3$ (e.g. $B_1$ with $bc$) contradicting that $b_1,b_2,b_3$
are false-twin vertices.

\textbf{Case 1.1}: One of the bicliques, say $B_1$, is isomorphic to $K_{1,r}$ where the vertex $a$ is in the partition of size one.
As the biclique containing $bc$ must intersect $B_1$, there exists a vertex $d \in B_1$ adjacent to
$b$ and not adjacent to $c$. Now, as $c \notin B_1$, there exists a vertex $e \in B_1$, such that $c$ is adjacent to $e$.
Therefore $\{a,b,c,d,e\}$ induces a $gem$ or a $rocket$ depending on the
edge $eb$. See Figure~\ref{case1-1}.

\FloatBarrier
\begin{figure}[h]
	\centering
	\includegraphics[scale=.4]{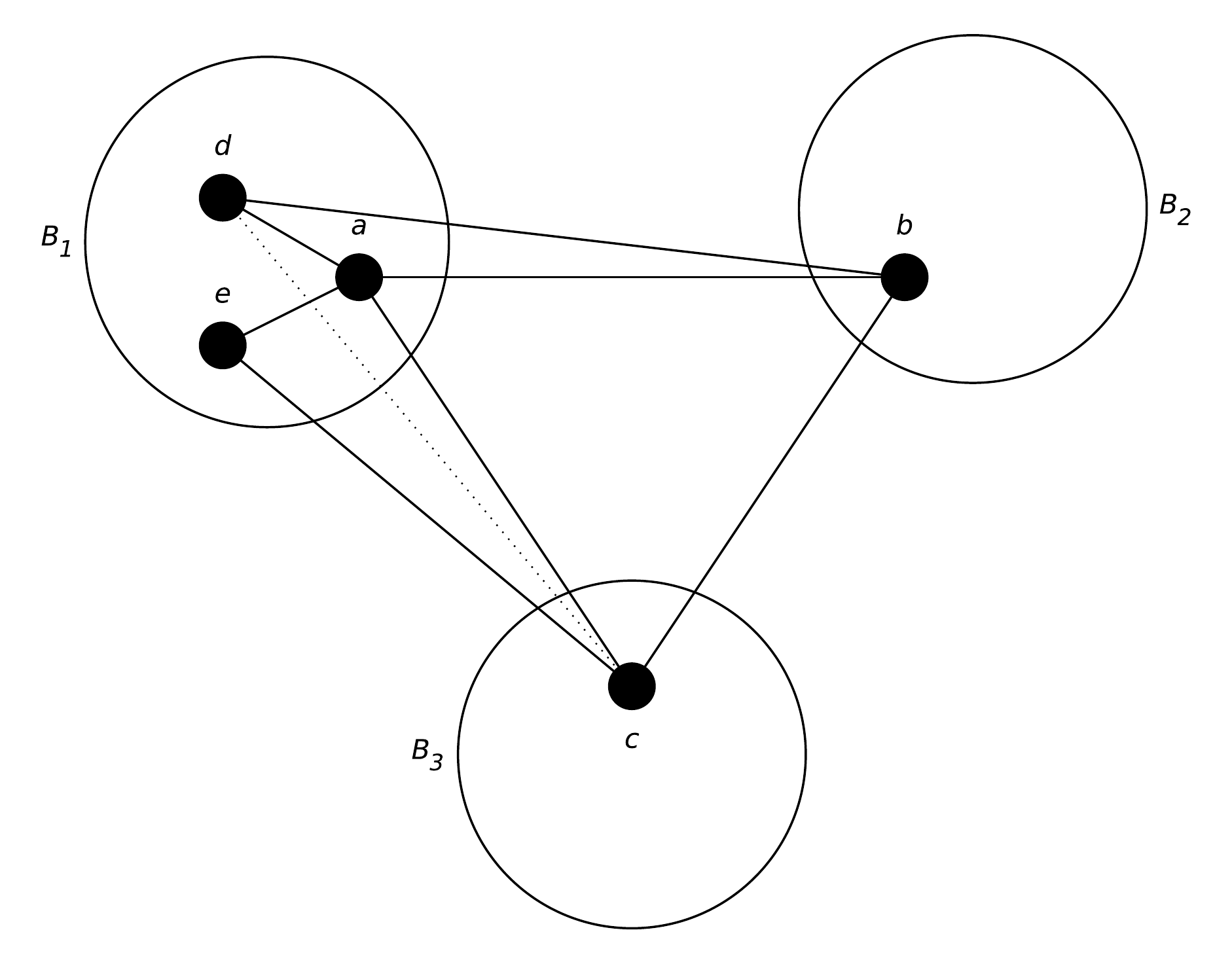}
	\caption{Case 1.1}
	\label{case1-1}
\end{figure}
\FloatBarrier

\textbf{Case 1.2}: None of the bicliques $B_1$, $B_2$ and $B_3$ are isomorphic to $K_{1,r}$ where the vertex of the $K_3$ is in the partition of 
size one. As the biclique containing $bc$ has to intersect $B_1$, 
call $e \in B_1$ a vertex in that intersection and w.l.g. assume $e$
adjacent to $c$ and not to $b$.

\textbf{Case 1.2.1}: Suppose $e$ is adjacent to $a$. Now, as $B_1$ is not isomorphic to $K_{1,r}$, we have the following cases. 

If there exists a vertex $g \in B_1$ adjacent to $e$ and not adjacent to
$b$.  Depending on the edge  $gc$,  $\{a,b,c,e,g\}$  induces  a  $gem$  or
$\{a,b,e\}$,  $\{b,c,e\}$, $\{a,c\}$  and $\{g,e,c\}$  are  contained in
four mutually intersecting bicliques.
See Figure~\ref{case1-2-1-1}.

\FloatBarrier
\begin{figure}[h]
	\centering
	\includegraphics[scale=.4]{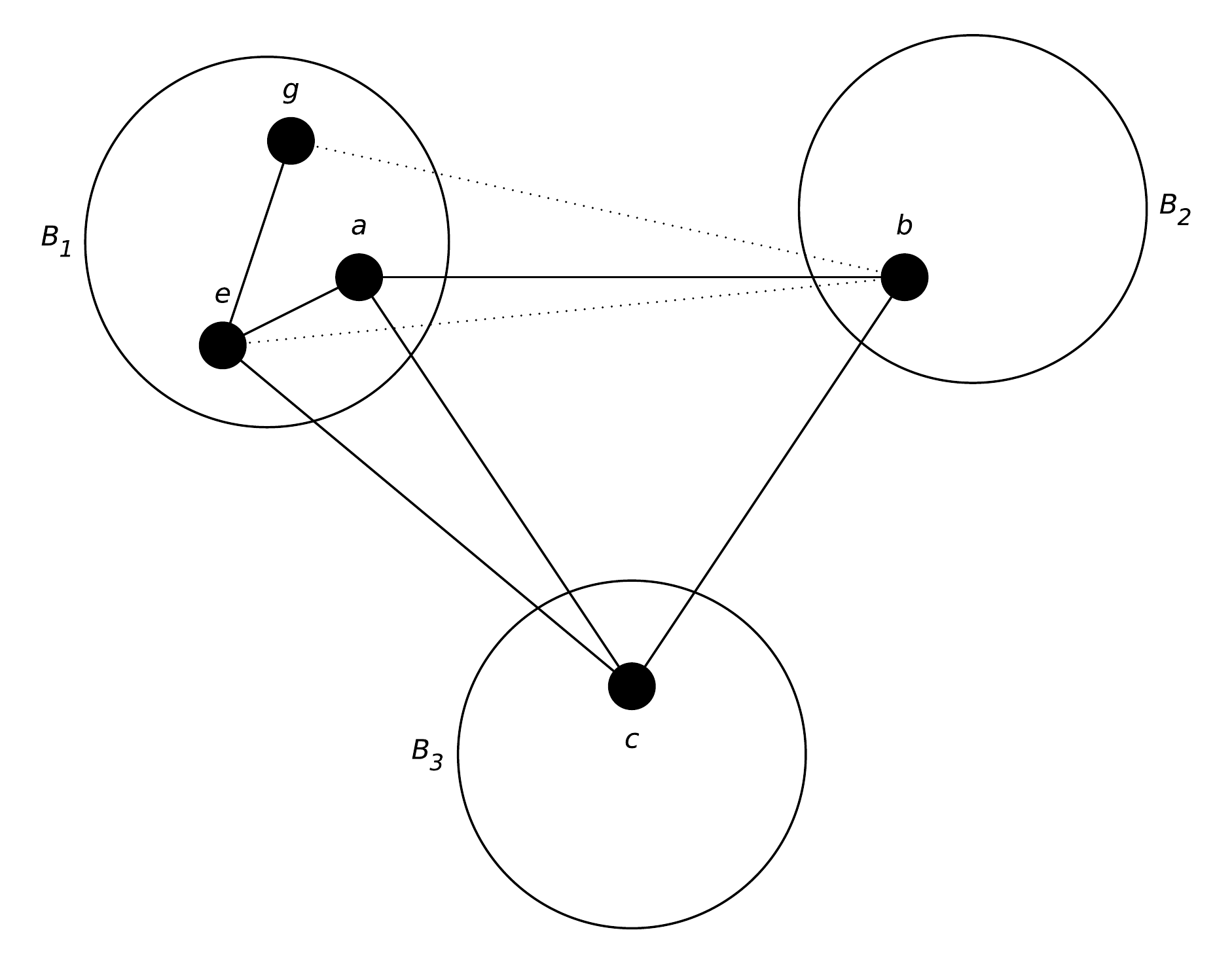}
	\caption{Case 1.2.1 with $g$ adjacent to $e$ and not to $b$}
	\label{case1-2-1-1}
\end{figure}
\FloatBarrier

Otherwise, assuming that  every $g \in B_1$ adjacent  to $e$ is adjacent
to $b$,  and considering that $b  \notin B_1$, there exists  $f \in B_1$
adjacent to $a$  and $b$.  In this case  $\{a,b,c,e,f\}$ induces a $gem$
or a $rocket$ depending on the edge $fc$.
See Figure~\ref{case1-2-1-2}.

\FloatBarrier
\begin{figure}[h]
	\centering
	\includegraphics[scale=.4]{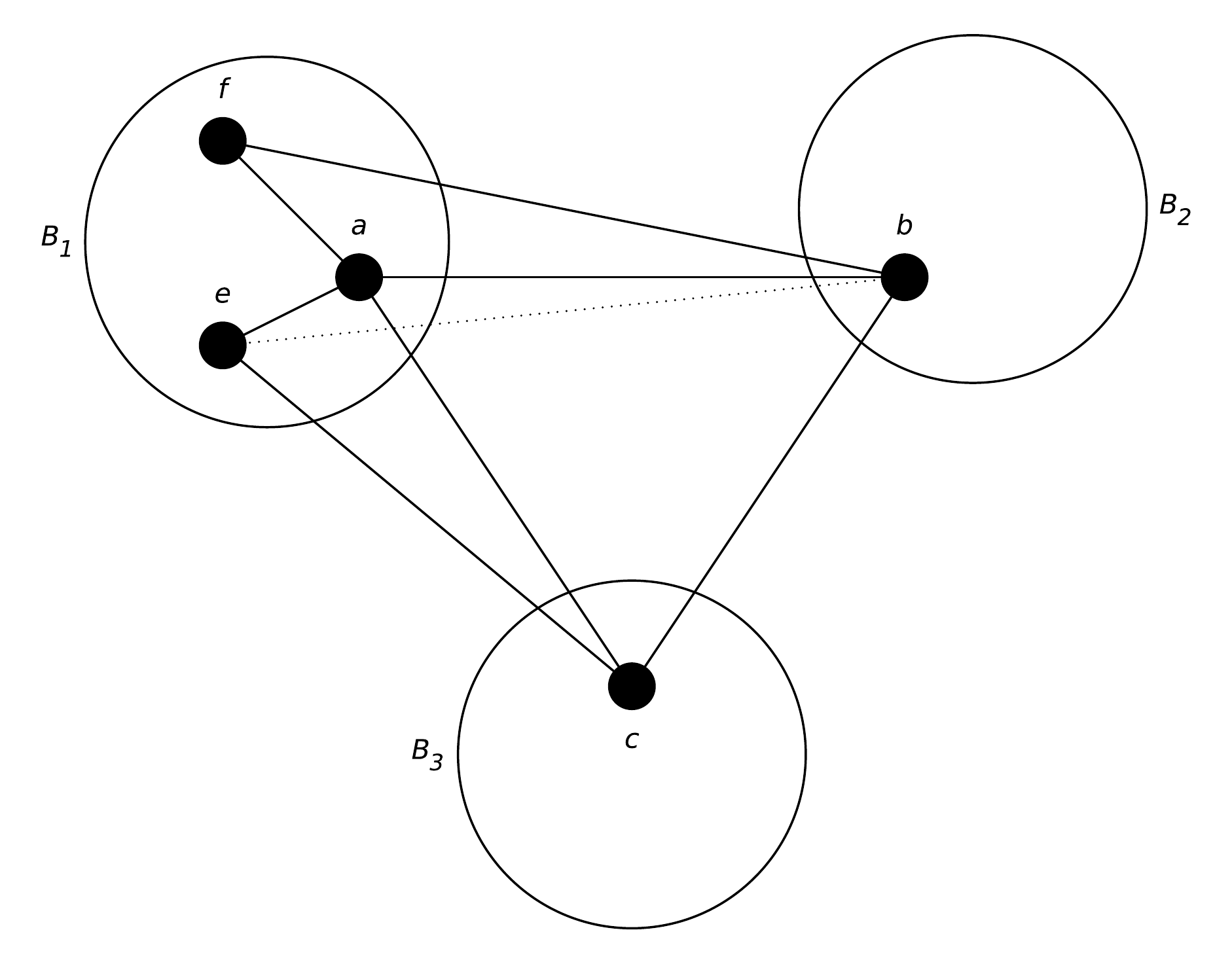}
	\caption{Case 1.2.1 with $f$ adjacent to $a$ and $b$}
	\label{case1-2-1-2}
\end{figure}
\FloatBarrier

\textbf{Case 1.2.2}: There exists $e \in B_1$ not adjacent to $a$ and $b$, and adjacent to $c$.
Let $h \in B_1$ be any vertex adjacent to $a$ (and consecuently to $e$). Clearly, if $h$ is adjacent to $c$, it must be adjacent to $b$, otherwise we
would be in the case above. So, if $h$ is adjacent to both, $\{a,b,c,e,h\}$ induces a $rocket$.
Therefore, we can assume that for every $h \in B_1$ adjacent to $e$ and $a$, $h$ is not adjacent to $b$ and $c$.
Moreover, this must be also true for every vertex in $B_2$ adjacent to $b$ and every vertex in $B_3$ adjacent to $c$, that is,
every vertex in $B_2$ adjacent to $b$ is not adjacent to $a$ and $c$, and every vertex in $B_3$ adjacent to $c$ is not 
adjacent to $a$ and $b$. Suppose that there exists $k\in B_2$ adjacent to $b$ and not adjacent to $h$, then $\{a,b,h\}$,
$\{a,b,k\}$, $\{b,c\}$ and $\{a,c\}$ are contained in four mutually intersecting bicliques.
Then, we can assume $k$ is adjacent to $h$. Indeed, assume that every vertex in $B_1$ adjacent to $a$ is adjacent to 
every vertex in $B_2$ adjacent to $b$ and to every vertex in $B_3$ adjacent to $c$. Also 
every vertex in $B_2$ adjacent to $b$ is adjacent to every vertex in $B_3$ adjacent to $c$.
Otherwise, we would obtain four mutually intersecting bicliques. Let $j \in B_3$ adjacent to $c$.
Observe that if $e$ is adjacent to $k$ then $e$ is also adjacent to $j$,
otherwise we are in case  1.2.1 considering the $K_3 = \{h,k,j\}$. Then,
depending on the edge $ek$,  $\{e,h,k,j,c\}$  induces  a $rocket$,  or
$\{a,b,k,h\}$, $\{a,c,j,h\}$, $\{b,c,k,j\}$ and $\{e,c,b\}$ 
are contained in four mutually intersecting bicliques. 
See Figure~\ref{case1-2-2}.

\FloatBarrier
\begin{figure}[h]
	\centering
	\includegraphics[scale=.4]{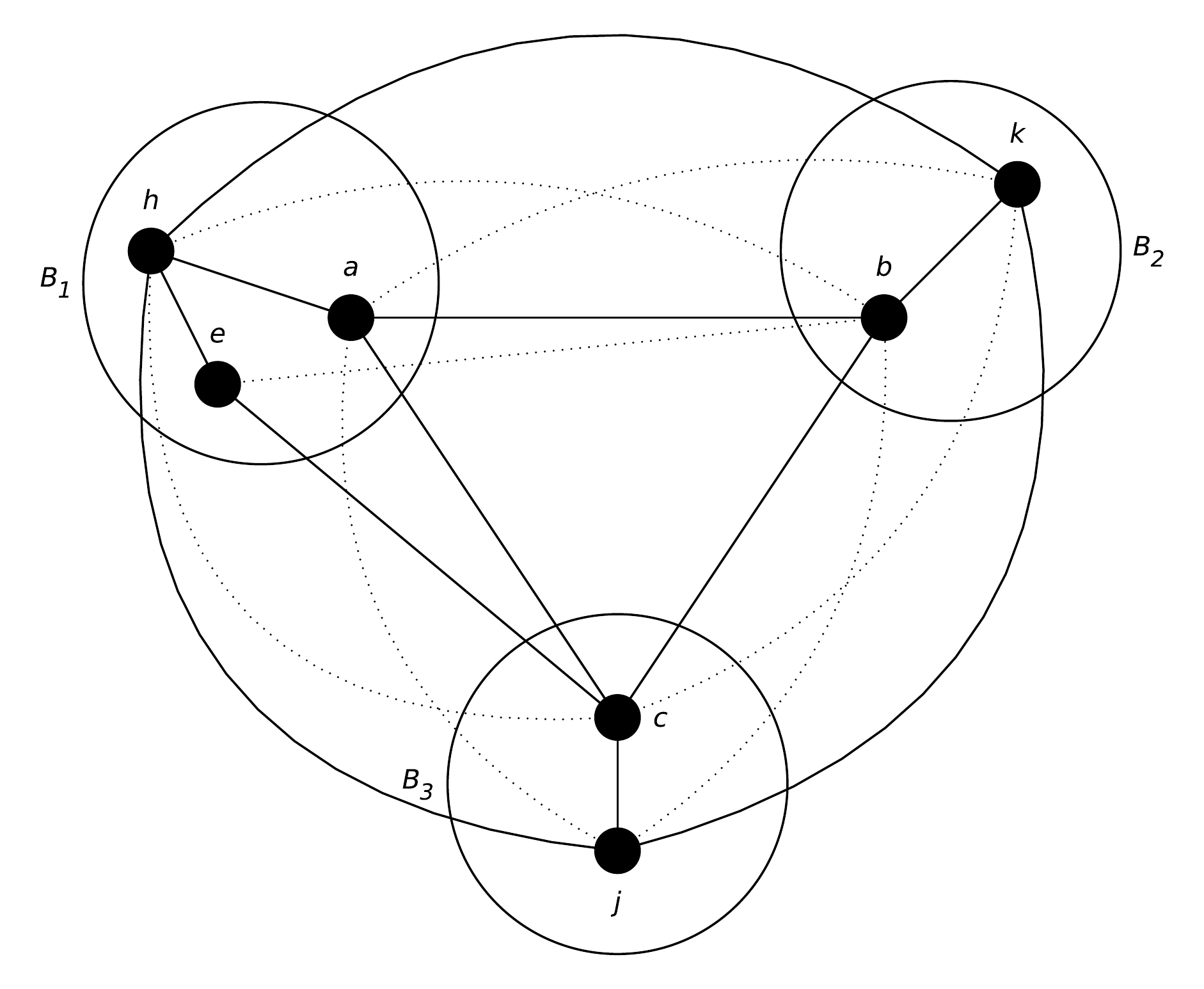}
	\caption{Case 1.2.2}
	\label{case1-2-2}
\end{figure}
\FloatBarrier
We covered all the cases when a $K_3$ is in $H$.

\textbf{Case 2}:  There is an induced $C_4=\{a,b,c,d\}$ in $H$ such
that $a,b\in  B_1$, $c\in B_2$ and $d\in  B_3$, that is, $ab,  bc, cd, ad
\in  E(H)$.  Now  as $c  \notin B_1$,  there exists  either $e  \in B_1$
adjacent to $b$ and $c$, or $h \in B_1$ adjacent to $a$ and not adjacent
to $c$. We have the following cases:

\textbf{Case 2.1}: $e$ is adjacent to $b$ and $c$ (the case where $e$ is
adjacent to $a$ and $d$ is analogous).  Observe that $e$ is not adjacent
to $d$ as  we would obtain a  triangle with one vertex  in each biclique
(case 1).  Let  $k \in  B_3$ be  a vertex  adjacent to  $d$. If  $k$ is
adjacent to  $c$ then  $\{b,e,c,d,k\}$ induces a  $butterfly$ (otherwise
case 1, considering $b$ and $k$, or $e$ and $k$, adjacent vertices).  Then  assume every vertex $k  \in B_3$ adjacent to  $d$ is not
adjacent to $c$. Furthermore, if any vertex $j \in B_2$ adjacent to $c$, 
is also adjacent to $d$, then $\{e,b,c,d,j\}$ induces a $butterfly$, a $gem$ or a $rocket$ depending
on the edges $ej$, $bj$. Therefore we can assume that every vertex $j \in
B_2$ adjacent to $c$  is not adjacent to $d$.
See Figure~\ref{case2-1}.

\FloatBarrier
\begin{figure}[h]
	\centering
	\includegraphics[scale=.4]{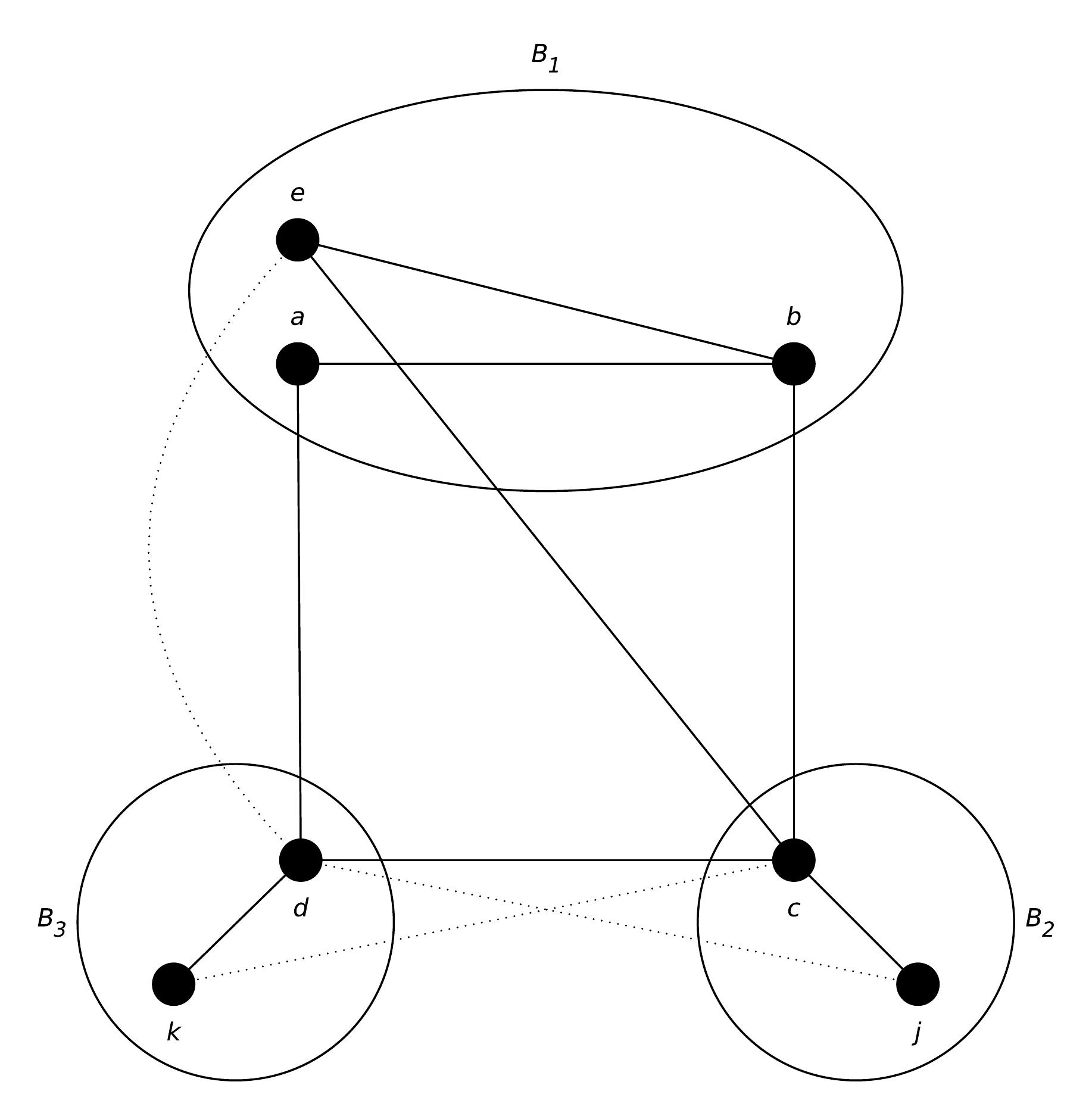}
	\caption{Case 2.1}
	\label{case2-1}
\end{figure}
\FloatBarrier

\textbf{Case 2.1.1}: There  is some $k$ not adjacent to  $b$.  Now as $c
\notin B_3$, there  exists $\ell \in B_3$ adjacent to  $k$ and not adjacent
to  $c$.  If  $\ell$  is adjacent  to  $b$ then  $\{\ell,b,c,d,k\}$ induces  a
$C_5$. We can assume $\ell$ is not adjacent to $b$.  

If $k$ is adjacent to $a$ then $\{a,b,c,d\}$, $\{a,b,k\}$, $\{c,d,k\}$ and
one of $\{a,k,\ell\}$ or $\{a,d,\ell\}$ depending on the edge $al$, are contained
in four different mutually intersecting bicliques. So we can assume $k$
is not adjacent to $a$.

As $a \notin B_2$,  either $a$ is not adjacent to some vertex of
$B_2$ that is adjacent to $c$, or $a$ forms a triangle with two vertices of
$B_2$. 

Suppose first that $a$ is not adjacent to $j \in B_2$  such that $j$ is
adjacent do $c$.  Note that $\{a,b,c,d\}$, $\{a,c,d,k\}$ and $\{c,d,e\}$
are contained  in three  different mutually intersecting  bicliques. See
Figure~\ref{case2-1-1a}.

\FloatBarrier
\begin{figure}[h]
	\centering
	\includegraphics[scale=.4]{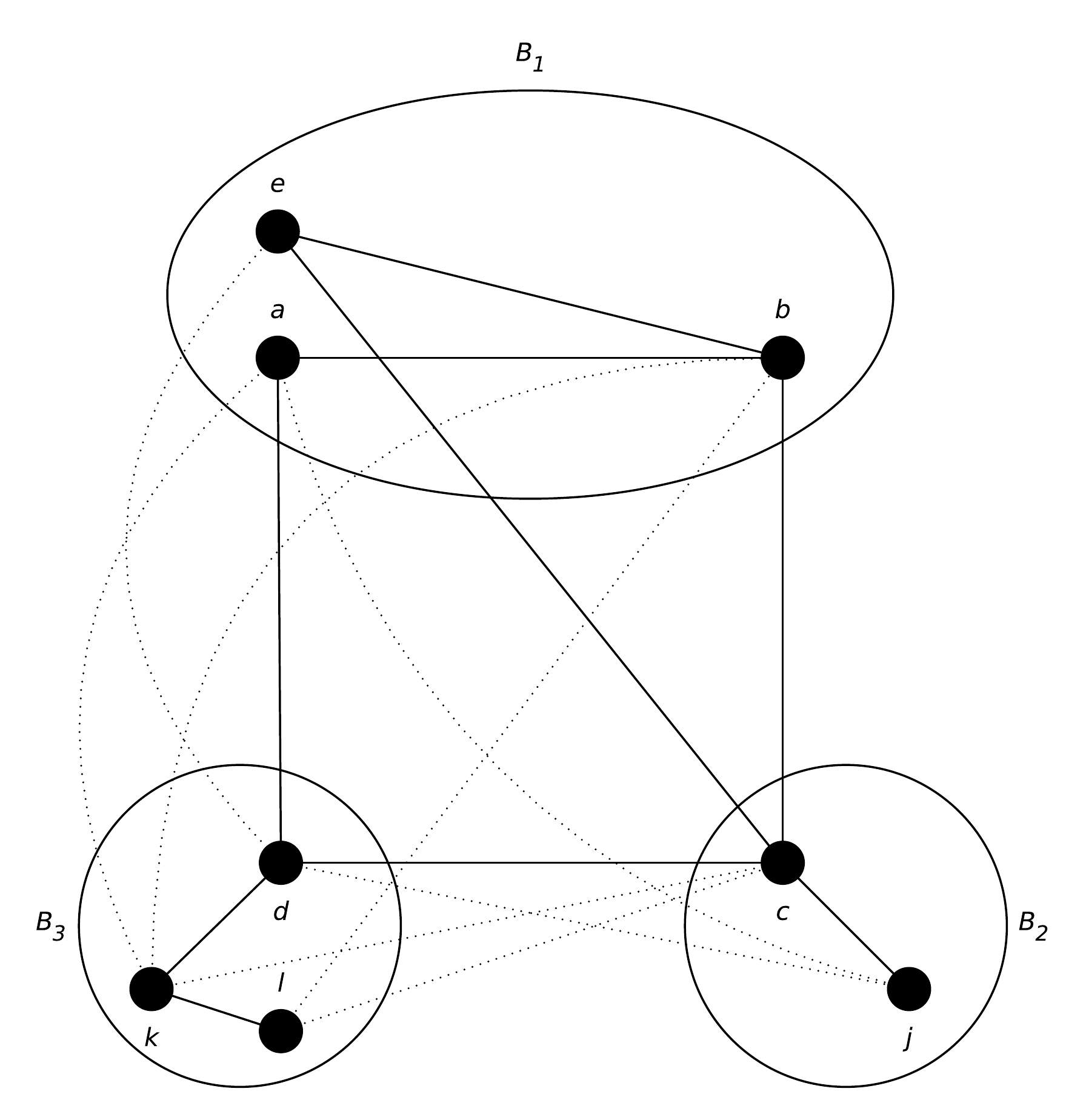}
	\caption{Case 2.1.1 with $a$ not adjacent to $j$}
	\label{case2-1-1a}
\end{figure}
\FloatBarrier

If $j$ is  not adjacent  to $b$  then $\{b,c,j\}$  is contained  in the
fourth  biclique  (and  we  got  four  different  mutually  intersecting
bicliques).  So suppose  $j$ is adjacent to $b$. If  $j$ is not adjacent
to $e$, the fourth biclique  contains $\{a,b,e,j\}$.  Finally, if $j$ is adjacent
to $e$ then $\{c,d,j\}$  is contained  in the
fourth biclique.

Suppose next that $a$ forms a triangle with two vertices of $B_2$. That is, there are two adjacent vertices  $j, p \in
B_2$ such that $j$ is adjacent to $c$ and $a$, and $p$ is adjacent
to $a$ (see Figure~\ref{case2-1-1b}).  If $p$ is adjacent to $b$, 
then depending on the  edge  $ep$,  $\{a,b,c,e,p\}$ induces a $butterfly$ or a 
$gem$. Assume therefore that $p$  is  not  adjacent  to  $b$. 
Then, $\{a,b,c,d\}$, $\{a,c,d,k\}$ and depending on the edge $dp$, either $\{c,d,e\}$ and $\{a,d,p\}$, 
or $\{c,d,j,p\}$  and  $\{a,b,p\}$  are  contained  in  four  different
mutually  intersecting  bicliques.  

\FloatBarrier
\begin{figure}[h]
	\centering
	\includegraphics[scale=.4]{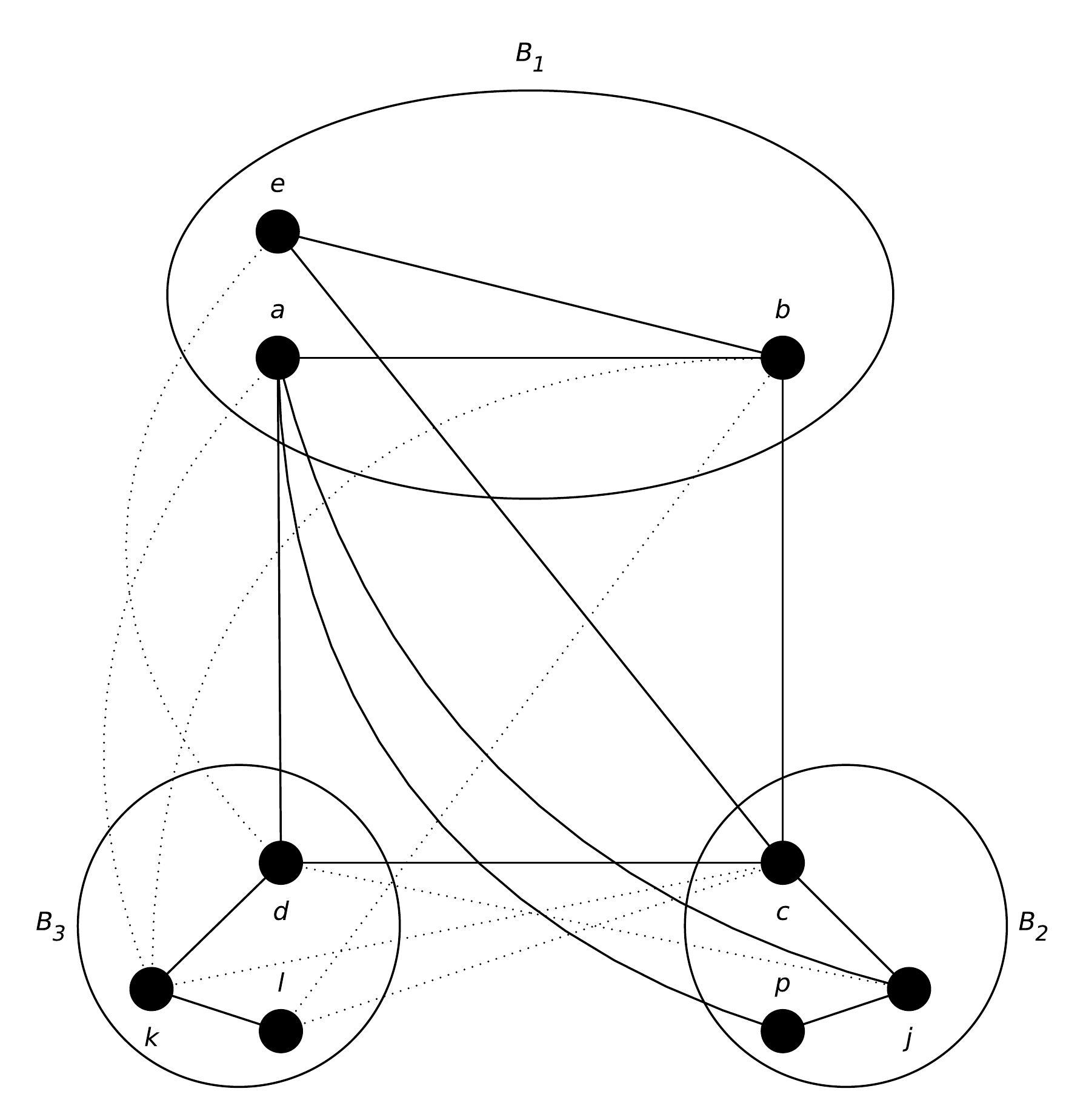}
	\caption{Case 2.1.1 $a$ form a  triangle with 2 vertices of $B_2$}
	\label{case2-1-1b}
\end{figure}
\FloatBarrier

\textbf{Case  2.1.2}: Every vertex  $k \in B_3$  adjacent to $d$  is
adjacent  to $b$.   Now as  $b \notin  B_3$, there  exists $m  \in B_3$
adjacent  to  $k$ and  $b$.  Note  that $m$  is  not  adjacent to  $c$,
otherwise case 1.  Then $\{b,e,c,k,m\}$ induces a $butterfly$, $gem$ or
$rocket$ depending on the edges $ek$ and $em$. See Figure~\ref{case2-1-2}.

\begin{figure}[ht!]
	\centering
	\includegraphics[scale=.4]{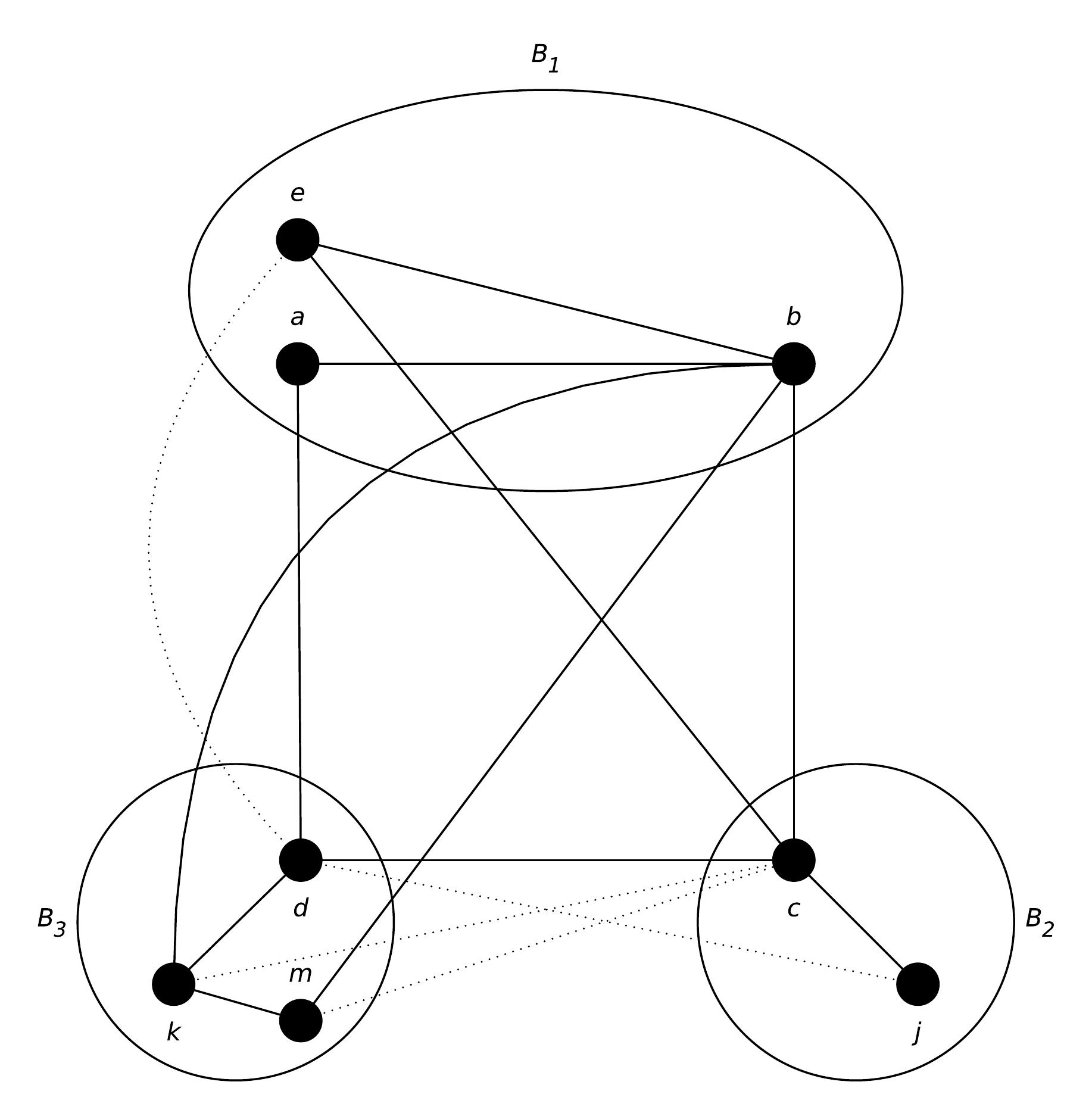}
	\caption{Case 2.1.2}
	\label{case2-1-2}
\end{figure}

\textbf{Case 2.2}:  $h$ is adjacent to  $a$ and not adjacent  to $c$. By
symmetry there  exists $g \in B_1$  adjacent to $b$ and  not adjacent to
$d$.  Assume that $g$ is not adjacent  to $c$ and $h$ is not adjacent to
$d$ (otherwise case 2.1).

Suppose first  that there  exists $k\in  B_3$ adjacent  to $d$  and $c$.
Observe  that  $k$ is  not  adjacent  to  $b$  (case 1  considering  the
$K_3=\{b,c,k\}$) and $k$ is  not adjacent to $a$ and to  $h$ at the same
time (case 2.1 considering the $C_4=\{b,c,k,a\}$). Depending on the edge
$ak$,  one  of  $\{b,c,k\}$  or $\{a,h,k\}$  along  with  $\{a,b,c,g\}$,
$\{a,b,c,d\}$, $\{a,b,d,h\}$  are contained  in four  different mutually
intersecting bicliques.

\begin{figure}[ht!]    
	\centering
	\includegraphics[scale=.4]{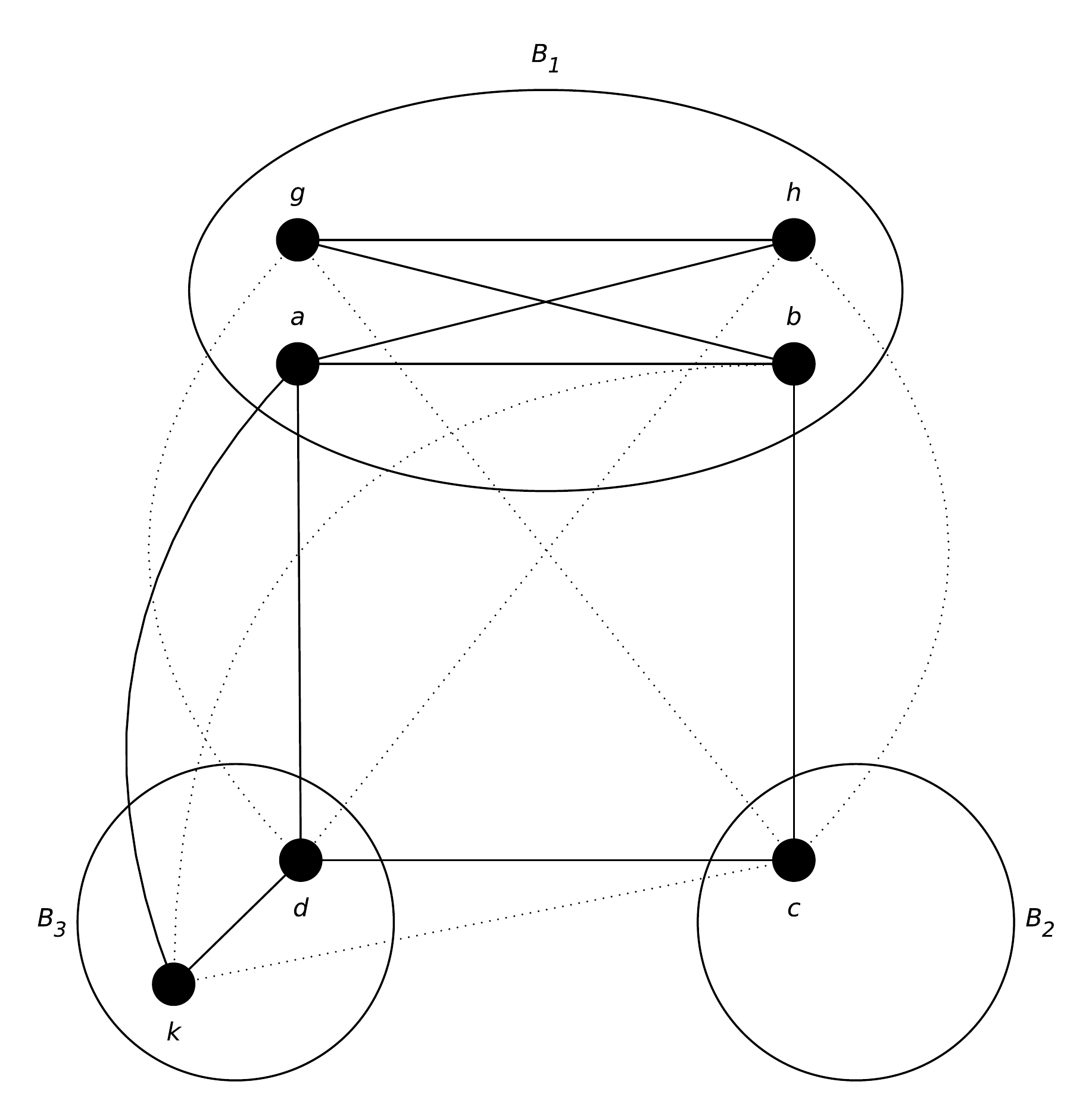}
	\caption{Case 2.2, with edge $ak$ and without edge $bk$}
	\label{case2-2a}
\end{figure}

Suppose therefore that every $k \in B_3$ adjacent to $d$ is not adjacent
to $c$.  If $k$ is not adjacent to $b$ or $k$ is  adjacent to $a$, then
$\{a,b,c,g\}$,  $\{a,b,c,d\}$,  $\{a,b,d,h\}$  and   $\{c,d,k\}$   are
contained  in  four  different  mutually  intersecting  bicliques.   See
Figure~\ref{case2-2a}. Otherwise, $k$ is adjacent  to $b$ and not  adjacent to $a$.
Consider the $C_4=\{k,d,c,b\}$, where the edge $kd$ is contained in
$B_3$, vertex $c\in B_2$ and $b\in B_1$.  Now, following the same
arguments as above, considering vertex $a$ as $k$, vertex $b$ as $d$, and
vertex $d$ as $b$, since the vertex $g\in B_1$ (that is adjacent to
$b$ and not adjacent to $d$ and $c$) has the same ``role'' as the vertex $k$, we arrive exactly to the previous case (when $k$ is not adjacent to 
$b$ or $k$ is adjacent to $a$, Figure~\ref{case2-2a}).
Therefore $\{k,d,c,g'\}$, $\{k,d,c,b\}$, $\{k,d,b,h'\}$ and $\{c,b,g\}$ are contained in four different mutually intersecting bicliques, 
where $h'\in B_3$ is adjacent to  $k$ and not adjacent to $c$ nor to $b$, and $g'\in B_3$ is
adjacent   to   $d$ and  not   adjacent   to   $d$  nor  to   $c$. See Figure~\ref{case2-2b}.

\begin{figure}[ht!]    
	\centering
	\includegraphics[scale=.4]{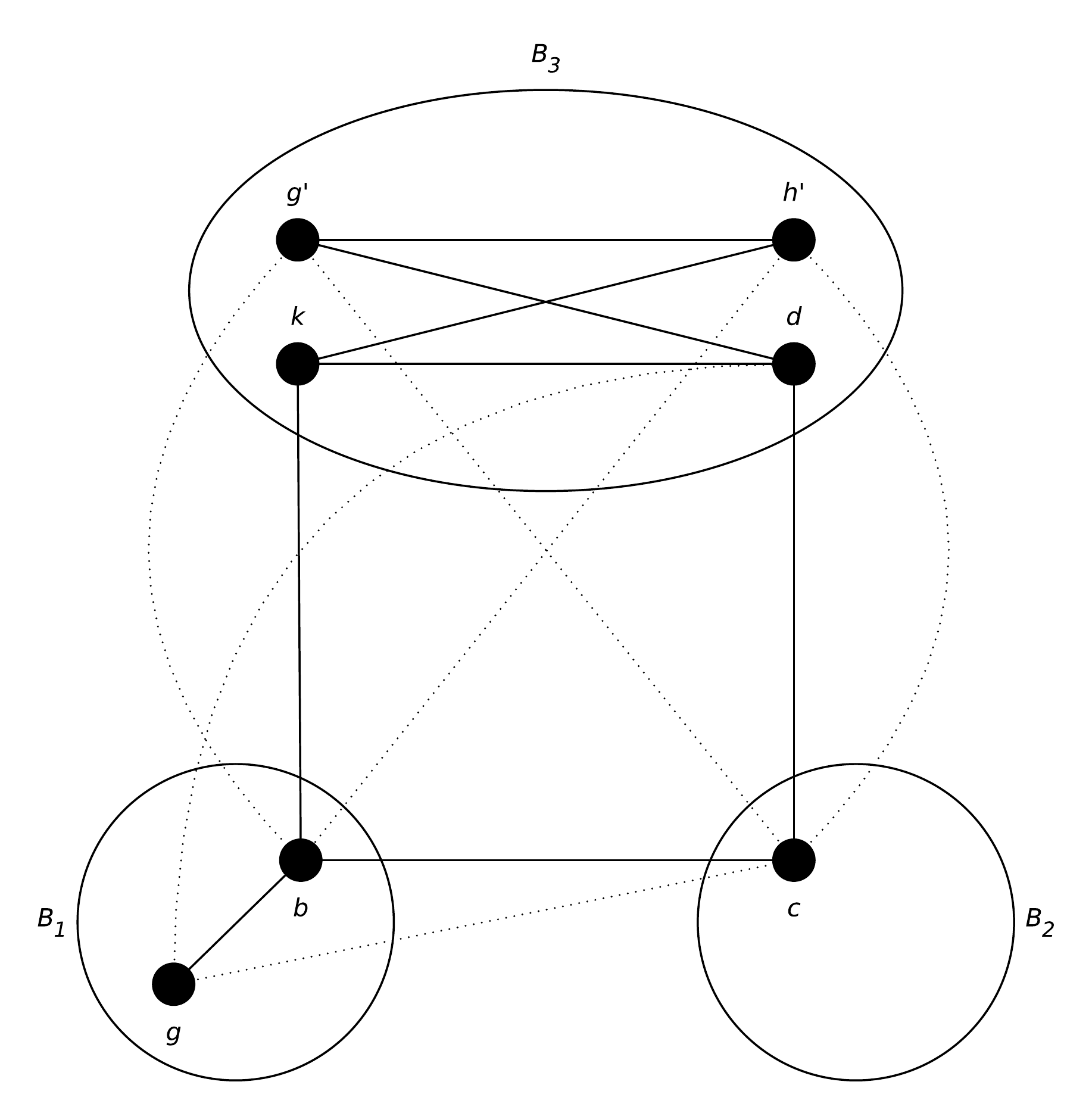}
	\caption{Case 2.2, with $C_4=\{k,d,c,b\}$}
	\label{case2-2b}
\end{figure}

We covered all the cases when a $C_4$ is in $H$ with all of the vertices
in the bicliques $B_1$, $B_2$ and $B_3$.

\textbf{Case 3}: There is an induced $C_k$, $5 \leq k \leq 9$ in $H$ with at least one vertex from each
biclique $B_1$, $B_2$ and $B_3$. For the case $k=5$ there is nothing to do.
Finally, for $6 \leq k \leq 9$, it is easy to see that, as each biclique containing two consecutive edges of the $C_k$ has
to intersect $B_1$, $B_2$ and $B_3$, then we would obtain a smaller cycle and therefore this case cannot occur.

Since we covered all cases the proof is done.

\end{proof}

Next, we present the main theorem of this section. This theorem shows that almost every graph is divergent under the biclique operator.
We remark that the linear time algorithm for recognizing convergent or divergent graphs given later in this section is based on this theorem.

\begin{theorem}
Let $G$ be a graph. If $G$ has at least $7$ bicliques, then $G$ diverges under the biclique operator.
\end{theorem}

\begin{proof} By way of contradiction, suppose that $G$ has at least $7$
bicliques and $G$ converges under the biclique operator. By
Corollary~\ref{contraccion}, $Tw(KB(G))=K_{n}$ for $n=1,...,4$.
Consider the following cases.

\textbf{Case $n=1$}. Then $KB(G)=K_1$ is a contradiction since $G$ has at least $7$
bicliques.

\textbf{Case $n=2$}. In~\cite{GroshausSzwarcfiterJGT2010} it was proved that no bipartite graph with more than two vertices is a biclique graph. Then $KB(G)=K_2$ what means that $G$ has only $2$ bicliques and therefore a contradiction.  

\textbf{Case $n=3$}. Since $G$ has at least $7$ bicliques it follows that in $KB(G)$
there exists a set of false-twin vertices of size at least three.  Consider the
bicliques $B_1,B_2,B_3$ of $G$ associated to the three false-twin vertices. If there
is a pair of bicliques $B_i,B_j$ such that there is no edge between any vertex
of $B_i$ and any vertex of $B_j$, by Lemma~\ref{2gemelos} it follows that $K_5$
is an induced subgraph of $KB(G)$. Otherwise, for every two pair of bicliques
$B_i, B_j$ there is an edge between some vertex of $B_i$ and some vertex of
$B_j$ and by Lemma~\ref{3gemelos} $KB(G)$ contains $K_5$ as an induced
subgraph. In any case, by Theorem~\ref{divergencia} $G$ diverges under the
biclique operator, a contradiction. 

\textbf{Case $n=4$}. There are two alternatives. Suppose that $KB(G)$ has a set of false-twin
vertices of size at least three. Then following the proof of the case $n=3$ we
arrive to a contradiction. Otherwise, there are only two possible graphs
isomorphic to $KB(G)$ ($KB(G)$ has $7$ or $8$ vertices and it has no set of
three false-twin vertices, see Fig.~\ref{Fig_7_8_vertices}). By inspection, using the characterization of biclique graphs given in~\cite{GroshausSzwarcfiterJGT2010},  
we prove that these two graphs are not biclique graphs. We conclude that this case cannot occur.

\begin{figure}[ht!]
	\centering
	\includegraphics[trim=0 100 0 100, clip, scale=.3]{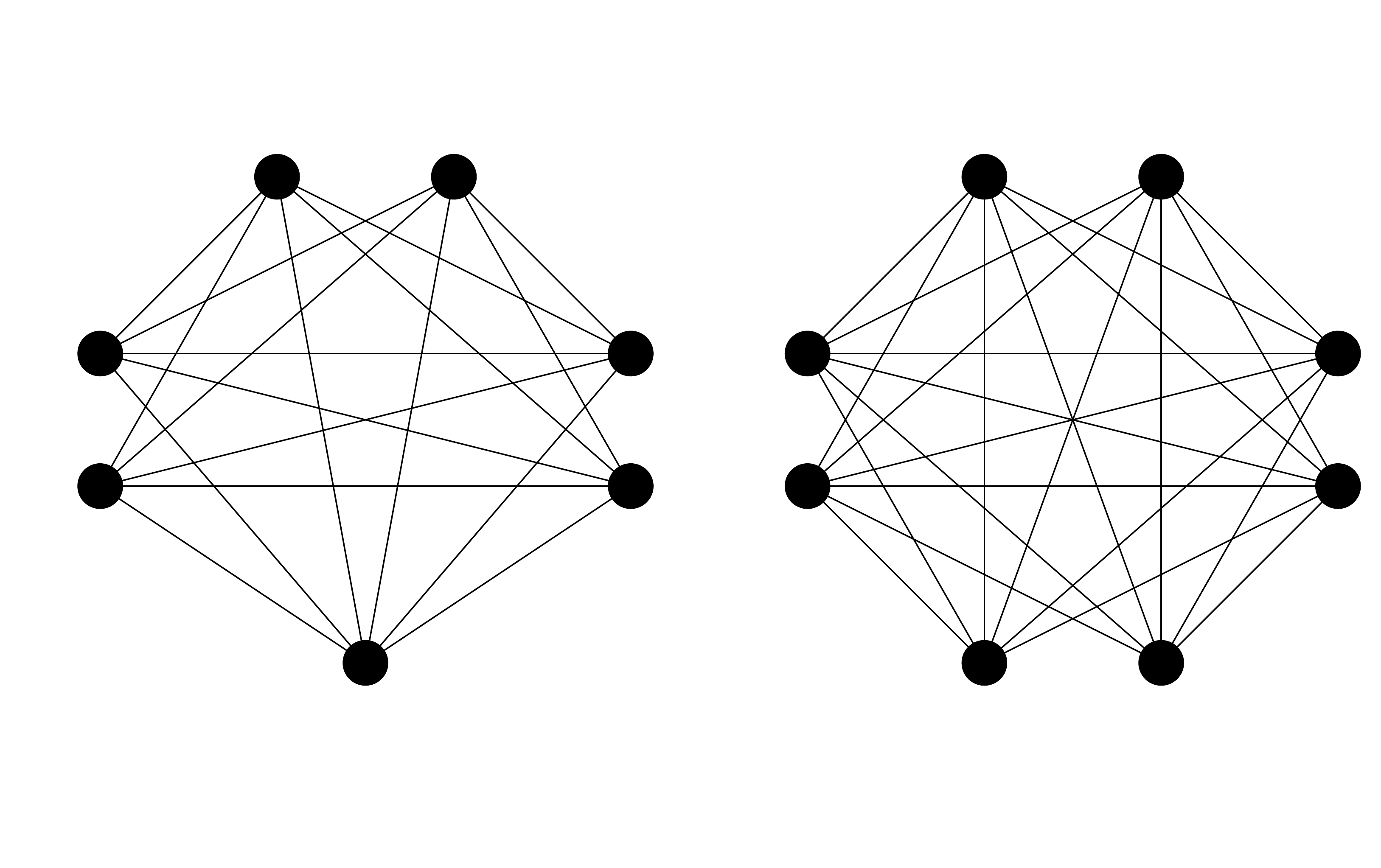}
	\caption{Unique two possible graphs for case $n=4$.}
	\label{Fig_7_8_vertices}
\end{figure}

Since we covered all cases, $G$ diverges under the biclique operator and the
proof is finished.
\end{proof}

The next step is to study graphs without false-twin vertices with at least $7$ bicliques.  This will complete the idea of  
the linear time algorithm for  recognizing divergent and convergent graphs under the biclique operator.

\begin{theorem}\label{13}
Let $G$ be a false-twin-free graph. If $G$ has at least $13$ vertices then $G$ has at least $7$ bicliques.
\end{theorem}
\begin{proof}
We prove the result by induction on $n$. For $n=13$, by inspection of all graphs without false-twin vertices the result holds.
Suppose now that $n\geq 14$. 
Theorem $2$ in~\cite{Sumner} states that if a graph $G$ has no false-twin vertices, then there exists a vertex $v$ such that $G-\{v\}$ is also 
false-twin free. Consider such a vertex $v$ and let $G' = G-\{v\}$. If $G'$ is connected, since it has at least $13$ vertices, by the inductive hypothesis it has at least $7$ bicliques.
Now as $G'$ is an induced subgraph of $G$ we conclude that $G$ also has at least $7$ bicliques. Suppose now that $G'$ is not connected.
Let $G_1,G_2,\ldots,G_s$ be the connected components of $G'$ on $n_1,n_2,\ldots,n_s$ vertices respectively. Since $G$ has no 
false-twin vertices, it can be at most one $G_i$ such that $n_i=1$. If there is one component with at least $13$ vertices, then by the inductive 
hypothesis this component has at least $7$ bicliques and so does $G$. Therefore every component has at most $12$ vertices. Now, by inspection we can 
verify that every component $G_i$ (but maybe one with just $1$ vertex) has at least $\lceil \frac{n_i}{2} \rceil$ bicliques. Also, since $G'$ is 
disconnected, $v$ along with at least one vertex of each of the $s$ components is a biclique in $G$ isomorphic to $K_{1,s}$ that is lost in $G'$. 
Summing up and assuming the worst case, that is, there exists one $n_i = 1$ (suppose $i=s$) we obtain that the number of bicliques of $G$ is at least 
\begin{displaymath}
\bigg(\sum\limits_{i=1}^{s-1}\bigg\lceil \frac{n_i}{2} \bigg\rceil\bigg) + 1 \geq
\Big\lceil \frac{11}{2} \Big\rceil + 1 = 7
\end{displaymath}
as we wanted to prove. Now the proof is complete.
\end{proof}

Theorem~\ref{13} implies that the number of convergent graphs without false-twin vertices is finite since convergent graphs without false-twin 
vertices have at most $12$ vertices. This fact leads to the following linear time algorithm. 

\textbf{Algorithm}: Given a graph $G$, build $H=Tw(G)$. If $H$ has at least $13$ vertices, answer ``$G$ diverges'' and STOP.
Otherwise, build $Tw(KB(H))$. If $Tw(KB(H))$ has at most $4$ vertices answer ``$G$ converges'' and STOP. Otherwise, answer ``$G$ diverges'' and STOP.

The algorithm has $O(n+m)$ time complexity. For this observe that $H$ can be built in $O(n+m)$ time using the modular decomposition~\cite{modulardecomp}. 
Finally, if $H$ has at most $12$ vertices any further operation takes $O(1)$ time complexity.

\section{Conclusions}
In~\cite{marinayo} it is given an $O(n^4)$ time algorithm to recognize convergent and divergent graphs under the biclique operator. 
In this paper we prove that graphs without false-twin vertices with at least $13$ vertices diverge. This shows that ``almost every'' graph is 
divergent and as a direct consequence, we obtain a linear time algorithm for recognizing the behavior of a graph under the biclique operator. 
We remark that in contrast as the iterated clique operator, no polynomial time algorithm is known for recognizing any of its possible behaviors. 

\bibliography{biblio2}

\end{document}